\documentclass[11pt]{llncs}
\usepackage{amsmath,amssymb}
\usepackage{graphicx}
\usepackage[linesnumbered, vlined, ruled]{algorithm2e}
\usepackage{fullpage}
\usepackage{hyperref,cite}
\usepackage{lineno}
\usepackage[top=0.8in, bottom=0.9in, left=0.9in, right=0.9in]{geometry}
\def\calU{\mathcal{U}}
\def\calL{\mathcal{L}}

\begin{document}

\title{Dynamic Convex Hulls under Window-Sliding Updates\thanks{A preliminary version of this paper appeared in {\em Proceedings of the 18th Algorithms and Data Structures Symposium (WADS 2023)}. This research was supported in part by NSF under Grants CCF-2005323 and CCF-2300356.}}
\author{
Haitao Wang
}
\institute{
Kahlert School of Computing\\
University of Utah, Salt Lake City, UT 84112, USA\\
\email{haitao.wang@utah.edu}\\
}

\maketitle

\pagestyle{plain}
\pagenumbering{arabic}
\setcounter{page}{1}

\vspace{-0.2in}
\begin{abstract}
We consider the problem of dynamically maintaining the convex hull of
a set $S$ of points in the plane under the following special sequence of insertions and
deletions (called {\em window-sliding updates}): insert a point to the
right of all points of $S$ and delete the leftmost point of $S$. We propose an $O(|S|)$-space data
structure that can handle each update in $O(1)$ amortized time, such that
standard binary-search-based queries on the convex hull of $S$ can be answered in $O(\log h)$
time, where $h$ is the number of vertices of the convex hull of $S$, and the convex hull itself can be output in $O(h)$ time.
%Our data structure, albeit not trivial, is quite simple.
%A big advantage of our algorithm is its simplicity.
\end{abstract}
%\vspace{-0.1in}


%\newpage

%\vspace{-0.2in}
\section{Introduction}
\label{sec:intro}

As a fundamental structure in computational geometry, the convex hull $CH(S)$ of a set $S$ of points in the plane has been well studied in the literature. Several $O(n\log n)$ time algorithms are known for computing $CH(S)$, e.g., see~\cite{ref:deBergCo08,ref:PreparataCo85}, where $n=|S|$, and the time matches the $\Omega(n\log n)$ lower bound. Output-sensitive $O(n\log h)$ time algorithms have also been given~\cite{ref:ChanOp96,ref:KirkpatrickTh86}, where $h$ is the number of vertices of $CH(S)$. If the points of $S$ are already sorted, e.g., by $x$-coordinate, then $CH(S)$ can be computed in $O(n)$ time by Graham's scan~\cite{ref:GrahamAn72}.

Due to a wide range of applications, the problem of dynamically maintaining $CH(S)$, where points can be inserted and/or deleted from $S$, has also been extensively studied.
Overmars and van Leeuwen~\cite{ref:OvermarsMa81} proposed an $O(n)$-space data structure that can support each insertion and deletion in $O(\log^2 n)$ worst-case time; Preparata and Vitter~\cite{ref:PreparataA93} gave a simpler method with the same performance. If only insertions are involved, then the approach of Preparata~\cite{ref:PreparataAn79} can support each insertion in $O(\log n)$ worst-case time.
For deletions only, Hershberger and Suri's method~\cite{ref:HershbergerAp92} can support each update in $O(\log n)$ amortized time. If both insertions and deletions are allowed, a breakthrough was given by Chan~\cite{ref:ChanDy01}, who developed a data structure of linear space that can support each update in $O(\log^{1+\epsilon}n)$ amortized time, for an arbitrarily small $\epsilon>0$. Subsequently, Brodal and Jacob~\cite{ref:BrodalDy00}, and independently Kaplan et al.~\cite{ref:KaplanFa01} reduced the update time to $O(\log n\log\log n)$. Finally, Brodal and Jacob~\cite{ref:BrodalDy02} achieved $O(\log n)$ amortized time performance for each update, with $O(n)$ space.

Under certain special situations, better and simpler results are also known. If the insertions and deletions are given offline, the data structure of Hershberger and Suri~\cite{ref:HershbergerOf96} can support $O(\log n)$ amortized time update. Schwarzkopf~\cite{ref:SchwarzkopfDy91} and Mulmuley~\cite{ref:MulmuleyRa91} presented algorithms to support each update in $O(\log n)$ expected time if the sequence of updates is random in a certain sense. In addition, Friedman, Hershberger, and Snoeyink~\cite{ref:FriedmanEf96} considered the problem of maintaining the convex hull of a simple path $P$ such that vertices are allowed to be inserted and deleted from $P$ at both ends of $P$, and they gave a linear space data structure that can support each update in $O(\log |P|)$ amortized time (more precisely, $O(1)$ amortized time for each deletion and $O(\log |P|)$ amortized time for each insertion).
%Hershberger and Snoeyink presented an algorithm to had a result with sublogarithmic (log^*) amortized update time, when the main update operation is splitting along a simple polygonal path.
There are also other special dynamic settings for convex hulls, e.g.,~\cite{ref:ChanTw19,ref:HershbergerCa98}.

%(more precisely, each deletion takes $O(1)$ amortized time and each insertion takes $O(\log (1+d))$ time, where $d$ is the number of vertices of the convex hull of $P$ removed due to the insertion).

In most applications, the reason to maintaining $CH(S)$ is to perform queries on it efficiently.
As discussed in Chan~\cite{ref:ChanTh12}, there are usually two types of queries, depending on whether the query is {\em decomposable}~\cite{ref:BentleyDe79}, i.e., if $S$ is partitioned into two subsets, then the answer to the query for $S$ can be obtained in constant time from the answers of the query for the two subsets. For example, the following queries are decomposable: find the most extreme vertex of $CH(S)$ along a query direction; decide whether a query line intersects $CH(S)$; find the two common tangents to $CH(S)$ from a query point outside $CH(S)$, while the following queries are not decomposable: find the intersection of $CH(S)$ with a vertical query line or more generally an arbitrary query line.
%decide whether a query point lies in $CH(S)$.
It seems that the decomposable queries are easier to deal with. Indeed, most of the above mentioned data structures can handle the decomposable queries in $O(\log n)$ time each. However, this is not the case for the non-decomposable queries. For example, none of the data structures of~\cite{ref:BrodalDy02,ref:BrodalDy00,ref:ChanDy01,ref:FriedmanEf96,ref:KaplanFa01} can support $O(\log n)$-time non-decomposable queries. More specifically, Chan's data structure~\cite{ref:ChanDy01} can be modified to support each non-decomposable query in $O(\log^{3/2}n)$ time but the amortized update time also increases to $O(\log^{3/2}n)$. Later Chan~\cite{ref:ChanTh12} gave a randomized algorithm that can support each non-decomposable query in expected $O(\log^{1+\epsilon}n)$ time, for an arbitrarily small $\epsilon>0$, and the amortized update time is also $O(\log^{1+\epsilon}n)$.

Another operation on $CH(S)$ is to output it explicitly, ideally in $O(h)$ time. To achieve this, one usually has to maintain $CH(S)$ explicitly in the data structure, e.g., in~\cite{ref:HershbergerAp92,ref:OvermarsMa81}. Unfortunately, most other data structures are not able to do so, e.g.,~\cite{ref:BrodalDy02,ref:BrodalDy00,ref:ChanDy01,ref:FriedmanEf96,ref:HershbergerOf96,ref:KaplanFa01,ref:PreparataA93}, although they can output $CH(S)$ in a slightly slower $O(h\log n)$ time.
In particular, Bus and Buzer~\cite{ref:BusDy15} considered this operation for maintaining the convex hull of a simple path $P$ such that vertices are allowed to be inserted and deleted from $P$ at both ends of $P$, i.e., in the same problem setting as in \cite{ref:FriedmanEf96}. Based on a modification of the algorithm in~\cite{ref:MelkmanOn87}, they achieved $O(1)$ amortized update time such that $CH(S)$ can be output explicitly in $O(h)$ time~\cite{ref:BusDy15}. However, no other queries on $CH(S)$ were considered in~\cite{ref:BusDy15}.

\subsection{Our results}

We consider a special sequence of insertions and deletions: the point inserted by an insertion must be to the right of all points of the current set $S$, and a deletion always happens to the leftmost point of the current set $S$. Equivalently, we may consider the points of $S$ contained in a window bounded by two vertical lines that are moving rightwards (but the window width is not fixed), so we call them the {\em window-sliding updates}.

To solve the problem, one can apply any previous data structure for arbitrary point updates. For example,  the method in~\cite{ref:BrodalDy02} can handle each update in $O(\log n)$ amortized time and answer each decomposable query in $O(\log n)$ time. Alternatively, if we connect all points of $S$ from left to right by line segments, then we can obtain a simple path whose ends are the leftmost and rightmost points of $S$, respectively. Therefore, the data structure of Friedman et al.~\cite{ref:FriedmanEf96} can be applied to handle each update in $O(\log n)$ amortized time and support each decomposable query in $O(\log n)$ time. In addition, although the data structure in~\cite{ref:HershbergerAp92} is particularly for deletions only, Hershberger and Suri~\cite{ref:HershbergerAp92} indicated that their method also works for the window-sliding updates, in which case each update (insertion and deletion) takes $O(\log n)$ amortized time. Further,
%as the data structure~\cite{ref:HershbergerAp92} explicitly stores the edges of $CH(S)$ in a balanced binary search tree, it
the data structure~\cite{ref:HershbergerAp92}
can support both decomposable and non-decomposable queries each in $O(\log n)$ time and report $CH(S)$ in $O(h)$ time.

In this paper, we provide a new data structure for the window-sliding updates. Our data structure uses $O(n)$ space and can handle each update in $O(1)$ amortized time. All decomposable and non-decomposable queries on $CH(S)$ mentioned above can be answered in $O(\log h)$ time each.\footnote{In the preliminary version of the paper at WADS 2023, the query time was $O(\log n)$. In this version, we improve the query time to $O(\log h)$.} Further, after each update, the convex hull $CH(S)$ can be output explicitly in $O(h)$ time. Specifically, the following theorem summarizes our result.

\begin{theorem}\label{theo:convexhull}
We can dynamically maintain the convex hull $CH(S)$ of a set $S$ of points in the plane to support each window-sliding update (i.e., either insert a point to the right of all points of $S$ or delete the leftmost point of $S$) in $O(1)$ amortized time, such that the following operations on $CH(S)$ can be performed. Let $n=|S|$ and $h$ be the number of vertices of $CH(S)$ right before each operation.
\begin{enumerate}
  \item The convex hull $CH(S)$ can be explicitly output in $O(h)$ time.
  \item
        Given two vertical lines, the vertices of $CH(S)$ between the vertical lines can be output in order along the boundary of $CH(S)$ in $O(k+\log h)$ time, where $k$ is the number of vertices of $CH(S)$ between the two vertical lines.
  \item Each of the following queries can be answered in $O(\log h)$ time.
      \begin{enumerate}
      \item\label{item:query}
      Given a query direction, find the most extreme point of $S$ along the direction.
      \item
      Given a query line, decide whether the line intersects $CH(S)$.
      \item
      Given a query point outside $CH(S)$, find the two tangents from the point to $CH(S)$.
      \item
      Given a query line, find its intersection with $CH(S)$, or equivalently, find the edges of $CH(S)$ intersecting the line.
      \item
      Given a query point, decide whether the point is in $CH(S)$.
      \item
      Given a convex polygon of $O(h)$ vertices (represented in any data structure that supports binary search), decide whether it intersects $CH(S)$, and if not, find their common tangents (both outer and inner).
      \end{enumerate}
\end{enumerate}
\end{theorem}

Comparing to all previous work, albeit on a very special sequence of updates, our result is particularly interesting due to the $O(1)$ amortized update time as well as its simplicity.

%A by-product of our technique is to dynamically maintain the convex hull of a set $Q$ of points for a special sequence of point insertions and deletions as in the above circular hull case, i.e., the point in each insertion must be to the right of all points of $Q$ and the point in each deletion must be the leftmost point of $Q$. We construct a data structure of $O(n)$ size (with $n=|Q|$) to maintain the convex hull $CH(Q)$ of $Q$ in  $O(1)$ amortized time per update such that the standard repertoire of binary-search-based queries on $CH(Q)$  can be implemented in $O(\log n)$ time each, e.g., given a point $q$, decide whether $q$ is in $CH(Q)$; given a line, compute the intersection of the line with $CH(Q)$; given a point outside $CH(Q)$, compute the two tangents from the point; given a direction, find the most extreme vertex of $CH(Q)$ along the direction, etc.

\paragraph{\bf Applications.}
Although the updates in our problem are quite special, the problem still finds applications. For example, Becker et al.~\cite{ref:BeckerAn91} considered the problem of finding two rectangles of minimum total area to enclose a set of $n$ rectangles in the plane. They gave an algorithm of $O(n\log n)$ time and $O(n\log\log n)$ space. Their algorithm has a subproblem of processing a dynamic set of points to answer queries of Type~\ref{item:query} of Theorem~\ref{theo:convexhull} with respect to window-sliding updates (see Section~3.2~\cite{ref:BeckerAn91}).
The subproblem is solved using subpath convex hull query data structure in~\cite{ref:GuibasCo91}, which costs $O(n\log \log n)$ space.
Using Theorem~\ref{theo:convexhull}, we can reduce the space of the algorithm to $O(n)$ while the runtime is still $O(n\log n)$. Note that Wang~\cite{ref:WangAl20} recently improved the space of the result of~\cite{ref:GuibasCo91} to $O(n)$, which also leads to an $O(n)$ space solution for the algorithm of~\cite{ref:BeckerAn91}. However, the approach of Wang~\cite{ref:WangAl20} is much more complicated.

Becker et al.~\cite{ref:BeckerEn96} extended their work above and studied the problem of enclosing a set of simple polygons using two rectangles of minimum total area. They gave an algorithm of $O(n\alpha(n)\log n)$ time and $O(n\log\log n)$ space, where $n$ is the total number of vertices of all polygons and $\alpha(n)$ is the inverse Ackermann function. The algorithm has a similar subproblem as above (see Section~4.2~\cite{ref:BeckerEn96}). Similarly, our result can reduce the space of their algorithm to $O(n)$ while the runtime is still $O(n\alpha(n)\log n)$.

We believe that our result may find other applications that remain
to be discovered.

\paragraph{\bf Outline.}
After introducing notation in Section~\ref{sec:pre}, we will prove
Theorem~\ref{theo:convexhull} gradually as follows. First, in
Section~\ref{sec:preconvexhull}, we give a data structure for a special problem
setting.
%in which we are given two sorted point sets $L$ and $R$ separated by a
%vertical line $\ell$ and $S$ is always a subsequence of $L\cup R$
%(with $S=L$ initially) such that insertions always happen to points in
%$R$ and deletions always happen to points in $L$. In fact, points of
%$L$ are given offline while points of $R$ are given online.
Then we
extend our techniques to the general problem setting in
Section~\ref{sec:generalconvexhull}. The data structures in
Section~\ref{sec:preconvexhull} and~\ref{sec:generalconvexhull} can only
perform the first operation in Theorem~\ref{theo:convexhull} (i.e., output $CH(S)$), we will
enhance the data structure in Section~\ref{sec:queries} so that other
operations can be handled. Section~\ref{sec:con} concludes with some
remarks.

\section{Preliminaries}
\label{sec:pre}

Let $\calU(S)$ denote the upper hull of $CH(S)$. We will focus on maintaining $\calU(S)$, and the lower hull can be treated likewise. The data structure for both hulls together will achieve Theorem~\ref{theo:convexhull}.

For any two points $p$ and $q$ in the plane, we say that $p$ is to the
{\em left} of $q$ if the $x$-coordinate of $p$ is smaller than or
equal to that of $q$. Similarly, we can define ``to the right of'',
``above'', and ``below''. We add ``strictly'' in front of them to
indicate that the tie case does not happen. For example, $p$ is
strictly below $q$ if the $y$-coordinate of $p$ is smaller than
that of $q$.

For a line segment $s$ and a point $p$, we say that $p$ is {\em vertically
below} $s$ if the vertical line through $p$ intersects $s$ at a point
above $p$ ($p\in s$ is possible). For any two line segments $s_1$ and $s_2$, we say that $s_1$ is vertically below $s_2$ if both endpoints of $s_1$ are vertically below $s_2$.

Suppose $\calL$ is a sequence of points and $p$ and $q$ are two points
of $\calL$. We will adhere to the convention that a subsequence of
$\calL$ {\em between $p$ and $q$} includes both $p$ and $q$, but a
subsequence of $\calL$ {\em strictly} between $p$ and $q$ does not include
either one.
%In many cases, $\calL$ is a cyclic sequence of points,
%e.g., vertices on $CH(S)$, and we often say points of $\calL$
%clockwise/counterclockwise (strictly) from $p$ to $q$.

For ease of exposition, we make a general position assumption that no two points of $S$ have the same $x$-coordinate and no three points are collinear.
%This assumption is made mainly to keep this exposition simple - the
%presented algorithm can be generalized (without loss in the asymptotic time
%and space complexity bounds) to other cases.

\section{A special problem setting with a partition line}
\label{sec:preconvexhull}

In this section we consider a special problem setting.
%and we will extend the techniques to the general setting in Section~\ref{sec:generalconvexhull}.
Specifically, let $L=\{p_1,p_2,\ldots,p_n\}$ (resp.,
$R=\{q_1,q_2,\ldots,q_n\}$) be a set of $n$ points sorted by
increasing $x$-coordinate, such that all points of $L$ are strictly to the
left of a known vertical line $\ell$ and all points of $R$ are strictly to the right of $\ell$.
We want to
maintain the upper convex hull $\calU(S)$ of a consecutive subsequence $S$ of
$L\cup R=\{p_1,\ldots,p_n,q_1,\ldots,q_n\}$, i.e., $S=\{p_i,p_{i+1},\ldots,p_n,q_1,q_2,\ldots,q_j\}$, with $S=L$ initially, such that a deletion will
delete the leftmost point of $S$ and an insertion will insert the point
of $R$ right after the last point of $S$. Further, deletions only happen to points of $L$, i.e., once $p_n$ is deleted from $S$, no deletion will happen.
Therefore, there are a total of $n$ insertions and $n$ deletions.

Our result is a data structure that can
support each update in $O(1)$ amortized time, and after each update we
can output $\calU(S)$ in $O(|\calU(S)|)$ time. The data structure can
be built in $O(n)$ time on $S=L$ initially.
%We can use the same approach to maintain the lower hull of $S$.
Note that the set $L$ is given offline because $S=L$ initially, but points of $R$ are given online.
We will extend the techniques to the general problem setting in Section~\ref{sec:generalconvexhull}, and the data structure will be enhanced in Section~\ref{sec:queries} so that other operations on $CH(S)$ can be handled.

\subsection{Initialization}

Initially, we construct the data structure on $L$, as follows. We run Graham's scan to process points of $L$ leftwards from $p_n$ to $p_1$. Each vertex $p_i\in L$
is associated with a stack $Q(p_i)$, which is empty initially.
Each vertex $p_i$ also has two pointers $l(p_i)$ and $r(p_i)$, pointing to its left and right neighbors respectively if $p_i$ is a vertex of the current upper hull.
Suppose we are processing a point $p_i$. Then, the upper hull of
$p_{i+1},p_{i+2},\ldots,p_n$ has already been maintained by a doubly
linked list with $p_{i+1}$ as the head.
To process $p_i$, we run Graham's scan to find a vertex $p_j$ of the current upper hull such that $\overline{p_ip_j}$ is an edge of the new upper hull. Then, we push $p_i$ into the stack $Q(p_j)$, and set $l(p_j)=p_i$ and $r(p_i)=p_j$.
%check whether $\overline{p'p}\cup \overline{pp_i}$ makes right turn at $p$
%if we move from $p'$ to $p_i$ following the two edges, where $p$ is
%the head vertex of the linked list and $p'$ is the next vertex of	$p$ in the list. If yes, we remove $p_{i+1}$ from the linked list 	and update the head to $p$. Next, we continue this until either	$\overline{p'p}\cup \overline{pp_i}$ makes left turn at $p$ or 	$p'$ does not exits (i.e., $p$ is the only vertex in the linked list). In either case, we push $p_i$ on top of the stack $S_p$ and	set $p_i$ as the head of the new list with $p$ as its next vertex 	by setting $F(p_i)=p$ ($F(p_i)$ is called the {\em forward	pointer} of $p_i$).
The algorithm is done after $p_1$ is processed.
%Note that if $p_i$ is in $S_p$, then $\overline{p_ip}$ is the leftmost edge of the convex hull of $\{p_i,\ldots,p_n\}$. As such, the algorithm runs in $O(n)$ time.

The stacks essentially maintain the left neighbors of the vertices of the historical upper hulls so that when some points are deleted in future, we can traverse leftwards from any vertex on the current upper hull after those deletions. More specifically, if $p_i$ is a vertex on the current upper hull, then the vertex at the top of $Q(p_i)$ is the left neighbor of $p_i$ on the upper hull.
In addition, notice that once the right neighbor pointer $r(p_i)$ is set during processing $p_i$, it will never be changed. Hence, in future if $p_i$ becomes a vertex of the current upper hull after some deletions, $r(p_i)$ is the right neighbor of $p_i$ on the current upper hull.
Therefore, we do not need another stack to keep the right neighbor of $p_i$.

The above builds our data structure for $\calU(S)$ initially when $S=L$. In what follows, we discuss the general situation when $S$ contains both points of $L$ and $R$. Let $S_1=S\cap L$ and $S_2=S\cap R$.
The data structure described above is used for maintaining $\calU(S_1)$. For $S_2$, we only use a doubly linked list to store its upper hull $\calU(S_2)$, and the stacks are not needed.
%(because deletions will not happen to the points of $S_2$).
In addition, we explicitly maintain the common tangent $\overline{t_1t_2}$ of the two upper hulls $\calU(S_1)$ and $\calU(S_2)$, where $t_1$ and $t_2$ are the tangent points on $\calU(S_1)$ and $\calU(S_2)$, respectively. We also maintain the leftmost and rightmost points of $S$. This completes the description of our data structure for $S$.

Using the data structure we can output $\calU(S)$ in $O(|\calU(S)|)$ time as follows. Starting from the leftmost vertex of $S_1$, we follow the right neighbor pointers until we reach $t_1$, and then we output $\overline{t_1t_2}$. Finally, we traverse $\calU(S_2)$ from $t_2$ rightwards until the rightmost vertex. In the following, we discuss how to handle insertions and deletions.

\subsection{Insertions}

Suppose a point $q_j\in R$ is inserted into $S$. If $j=1$, then this is the first insertion. We set $t_2=q_1$ and find $t_1$ on $\calU(S_1)$ by traversing it leftwards from $p_n$ (i.e., by Graham's scan).
%We call it the {\em initial tangent searching} procedure, which
This takes $O(n)$ time but happens only once in the entire algorithm (for processing all $2n$ insertions and deletions), so the amortized cost for the insertion of $q_1$ is $O(1)$.
In the following we consider the general case $j>1$.

\begin{figure}[t]
\begin{minipage}[t]{\textwidth}
\begin{center}
\includegraphics[height=1.7in]{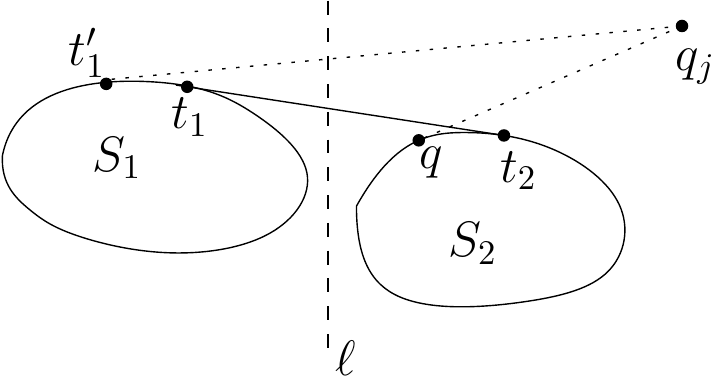}
\caption{\footnotesize Illustrating the insertion of $q_j$.}
\label{fig:convexhull-10}
\end{center}
\end{minipage}
\vspace{-0.15in}
\end{figure}

We first update $\calU(S_2)$ by Graham's scan. This procedure takes $O(n)$ time in total for all $n$ insertions, and thus $O(1)$ amortized time per insertion. Let $q$ be the vertex such that $\overline{qq_j}$ is the edge of the new hull $\calU(S_2)$ (e.g., see Fig.~\ref{fig:convexhull-10}). If $q$ is strictly to the right of $t_2$, or if $q=t_2$ and $\overline{t_1t_2}$ and $\overline{t_2q_j}$ make a right turn at $t_2$, then $\overline{t_1t_2}$ is still the common tangent and we are done with the insertion. Otherwise, we update $t_2=q_j$ and find the new $t_1$ by traversing $\calU(S_1)$ leftwards from the current $t_1$, and we call it the {\em insertion-type tangent searching procedure}, which takes $O(1+k)$ time, with $k$ equal to the number of vertices on $\calU(S_1)$ strictly between the original $t_1$ and the new $t_1$ (and we say that those vertices are {\em involved} in the procedure). The following lemma implies that the total time of this procedure in the entire algorithm is $O(n)$, and thus the amortized cost is $O(1)$.

\begin{lemma}\label{lem:insertioncost}
Each point of $L\cup R$ can be involved in the insertion-type
tangent searching procedure at most once in the entire algorithm.
\end{lemma}
\begin{proof}
We use $t_1$ to refer to the original tangent point and use $t_1'$ to
refer to the new one  (e.g., see Fig.~\ref{fig:convexhull-10}). Let
$p$ be any point on $\calU(S_1)$ strictly between $t_1'$ and $t_1$. Then, we have the following {\em observation}:
all points of $L$ strictly to the right of $p$
must be {\em vertically below} the segment $\overline{pq_j}$.
%(i.e., for each such point $p'$, the vertical line through $p'$ intersects $\overline{pq_j}$ and $p'$ is below the intersection).

Assume to the contrary that $p$ is involved again in the procedure
when another point $q_k$ with $k>j$ is inserted. Let $t_1''$ be the
tangent point on $\calU(S_1)$ right before $q_k$ is inserted. As $p$
is involved in the procedure, $p$ must be on the current
hull $\calU(S_1)$ and strictly to the left of $t_1''$.
According to the above observation, $t_1''$ is vertically below the segment $\overline{pq_j}$ (and $t_1''\not\in \overline{pq_j}$ due to the general position assumption).
Since $q_j$ is in the current set $S_2$ and $p$ is in the current set
$S_1$, and $t_1''$ is vertically below the segment $\overline{pq_j}$,
$t_1''$ cannot be the left tangent point of the common tangent of $\calU(S_1)$ and $\calU(S_2)$, incurring a
contradiction. \qed
\end{proof}

%This finishes our algorithm for processing the insertion of $p_j$.

\subsection{Deletions}

Suppose a point $p_i\in L$ is deleted from $S_1$.
If $i=n$, then this is the last deletion. In this case, we only need
to maintain $\calU(S_2)$ for insertions only in future, which can be done by
Graham's scan. In the following, we assume that $i<n$.

Note that $p_i$ must be the leftmost vertex of the current hull
$\calU(S_1)$. Let $p=r(p_i)$ (i.e., $p$ is the right neighbor of $p_i$ on $\calU(S_1)$).
According to our data structure, $p_i$ is at the top of the stack $Q(p)$.
We pop $p_i$ out of $Q(p)$. If $p_i\neq t_1$, then $p_i$ is strictly to the left of $t_1$ and
$\overline{t_1t_2}$ is still the common tangent of the new $S_1$ and $S_2$, and thus we are done with the
deletion. Otherwise, we find the new tangent by simultaneously
traversing on $\calU(S_1)$ leftwards from $p$ and traversing on
$\calU(S_2)$ leftwards from $t_2$ (e.g., see
Fig.~\ref{fig:convexhulldeletion}). Specifically, we first check
whether $\overline{pt_2}$ is tangent to $\calU(S_1)$ at $p$. If not,
then we move $p$ leftwards on the new $\calU(S_1)$ until
$\overline{pt_2}$ is tangent to $\calU(S_1)$ at $p$. Then, we check
whether $\overline{pt_2}$ is tangent to $\calU(S_2)$ at $t_2$. If not,
then we move $t_2$ leftwards on $\calU(S_2)$ until $\overline{pt_2}$
is tangent to $\calU(S_2)$ at $t_2$. If the new $\overline{pt_2}$ is
not tangent to $\calU(S_1)$ at $p$, then we move $p$ leftwards again.
We repeat the procedure until the updated $\overline{pt_2}$ is tangent
to $\calU(S_1)$ at $p$ and also tangent to $\calU(S_2)$ at $t_2$. Note
that both $p$ and $t_2$ are monotonically moving leftwards on
$\calU(S_1)$ and $\calU(S_2)$, respectively. Note also that traversing
leftwards on $\calU(S_1)$ is achieved by using the stacks associated with the vertices while traversing on $\calU(S_2)$ is done by using the doubly-linked list that stores $\calU(S_2)$.
We call the above the {\em deletion-type tangent searching procedure},
which takes $O(1+k_1+k_2)$ time, where $k_1$ is the number of points
on $\calU(S_1)$ strictly between $p$ and the new tangent point $t_1$, i.e., the final position of $p$ after the algorithm finishes (we say that these points are involved in the procedure), and $k_2$ is
the number of points on $\calU(S_2)$ strictly between the original
$t_2$ and the new $t_2$ (we say that these points are involved in the
procedure). The following lemma implies that the total time of this
procedure in the entire algorithm is $O(n)$, and thus the amortized
cost is $O(1)$.

\begin{figure}[t]
\begin{minipage}[t]{\textwidth}
\begin{center}
\includegraphics[height=1.7in]{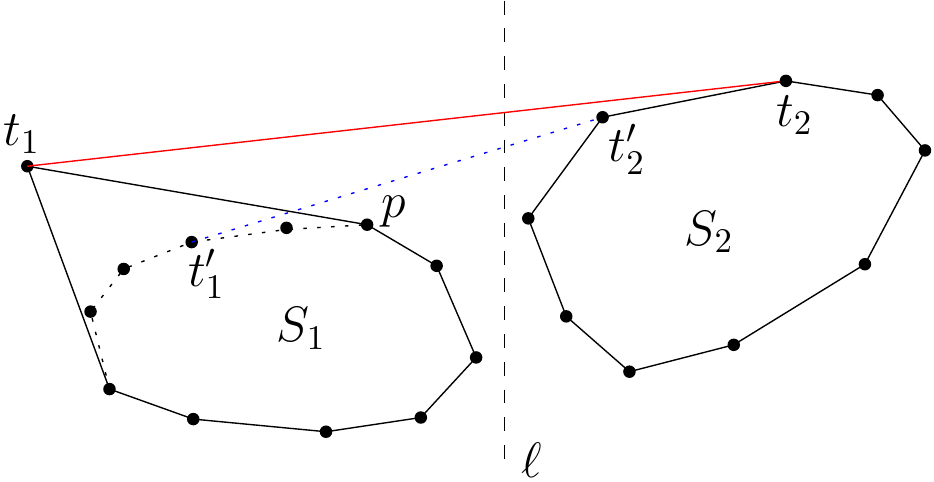}
\caption{\footnotesize Illustrating the deletion of $p_i$ where $p_i=t_1$. $\overline{t_1't_2'}$ are the new tangent of $\calU(S_1)$ and $\calU(S_2)$ after $p_i$ is deleted.}
\label{fig:convexhulldeletion}
\end{center}
\end{minipage}
\vspace{-0.15in}
\end{figure}

\begin{lemma}\label{lem:costconvexhull}
Each point of $L\cup R$ can be involved in the deletion-type
tangent searching procedure at most once in the entire algorithm.
\end{lemma}
\begin{proof}
Let $p'$ be a vertex on $\calU(S_1)$ strictly between $p$ and the new
tangent point $t_1$ (which is $t_1'$ in Fig.~\ref{fig:convexhulldeletion}). We argue that $p'$ cannot be involved in the same
procedure again in future. Indeed, $p'$ was not a vertex of $\calU(S_1)$
before $p_i$ is deleted because it was vertically below
the edge $\overline{p_ip}$ of $\calU(S_1)$. After $p_i$ is deleted, since $p'$ is involved in
the procedure, $p'$ must be  a vertex of $\calU(S_1)$. Further,
$p'$ will always be a vertex of $\calU(S_1)$ until it is deleted. Hence, $p'$
will never be involved in the procedure again in future (because to
do so $p'$ cannot be a vertex of $\calU(S_1)$ right before the deletion).

Let $t_2'$ be the new $t_2$ and $t_2$ be the original one (e.g., see
Fig.~\ref{fig:convexhulldeletion}). Let $q$ be
a vertex on $\calU(S_2)$ strictly between the two. We argue that $q$
will never be involved in the same procedure again in future.
Consider another deletion in future. Let $t_2''$ be the tangent point on
$\calU(S_2)$ right before the deletion. If $t_2''$ is strictly to the right of
$t_2'$, then one can verify that $q$ must have been removed from $\calU(S_2)$ due to the insertion of $t_2''$, and further, $q$ will never be a vertex of $\calU(S_2)$ again because only insertions will happen to
$S_2$. Hence, in this case $q$ cannot be involved in the procedure
(because to do so $q$ must be a vertex of $\calU(S_2)$).
If $t_2''$ is strictly to the left of $t_2'$ or at $t_2'$ (and thus is strictly
to the left of $q$), then the procedure
will search the new right tangent point by traversing on $\calU(S_2)$
leftwards from $t_2''$, and because $q$ is strictly to the right of $t_2''$, $q$
cannot be involved in the procedure.
\qed
\end{proof}

%This finishes our deletion procedure.

As a summary, in the special problem setting, we can perform each insertion and deletion in $O(1)$ amortized time, and after each update, the upper hull $\calU(S)$ can be output in $|\calU(S)|$ time.

%Note that in addition to that the insertions always happen to the right of all points of $S$ and a deletion also happens to the leftmost point of $S$, our result also has the following restrictions. First, insertions always happen to the right of $\ell$ and deletions only happen to the left of $\ell$. Second, as the upper hull $\calU(S)$ is essentially maintained by a linked list, we cannot perform any standard convex hull query in logarithmic time. In Section~\ref{sec:convexhull}, we will further extend our techniques to achieve the following result.

\section{The general problem setting}
\label{sec:generalconvexhull}

In this section, we extend our algorithm given in Section~\ref{sec:preconvexhull} to the general problem setting without the restriction on the existence of the partition line $\ell$. Specifically, we want to maintain the upper hull $\calU(S)$ under window-sliding updates, with $S=\emptyset$ initially.
We will show that each update can be handled in $O(1)$ amortized time and after each update $\calU(S)$ can be output in $O(|\calU(S)|)$ time.
We will enhance the data structure in Section~\ref{sec:queries} so that it can handle other operations on $CH(S)$.

During the course of processing updates, we maintain a vertical line $\ell$ that will move rightwards. At any moment, $\ell$ plays the same role as in the problem setting in Section~\ref{sec:preconvexhull}. In addition, $\ell$ always contains a point of $S$. Let $S_1$ be the subset of $S$ to the left of $\ell$ (including the point on $\ell$), and $S_2=S\setminus S_1$.
%Hence, the point contained in $\ell$ is in $S_2$, which is actually the leftmost point of $S_2$ by the definition of $S_2$.
For $S_1$, we use the same data structure as before to maintain $\calU(S_1)$, i.e., a doubly linked list for vertices of $\calU(S_1)$ and a stack associated with each point of $S_1$, and we call it {\em the list-stack data structure}. For $S_2$, as before, we only use a doubly linked list to store the vertices of $\calU(S_2)$. Note that $S_2=\emptyset$ is possible. If $S_2\neq \emptyset$, we also maintain the common tangent $\overline{t_1t_2}$ of $\calU(S_1)$ and $\calU(S_2)$, with $t_1\in \calU(S_1)$ and $t_2\in \calU(S_2)$. We can output the upper hull $\calU(S)$ in $O(|\calU(S)|)$ time as before.

For each $i\geq 1$, let $p_i$ denote the $i$-th inserted point. Let $U$ denote the universal set of all points $p_i$ that will be inserted. Note that points of $U$ are given online and we only use $U$ for the reference purpose in our discussion (our algorithm has no information about $U$ beforehand).
We assume that $S$ initially consists of two points $p_1$ and $p_2$. We let $\ell$ pass through $p_1$. According to the above definitions, we have $S_1=\{p_1\}$, $S_2=\{p_2\}$, $t_1=p_1$, and $t_2=p_2$.
We assume that during the course of processing updates $S$ always has at least two points, since otherwise we could restart the algorithm from this initial stage. Next, we discuss how to process updates.

\subsection{Deletions}
\label{subsec:deletegeneral}

Suppose a point $p_i$ is deleted. If $p_i$ is not the only point of $S_1$, then we do the same processing as before in Section~\ref{sec:preconvexhull}. We briefly discuss it here. If $p_i\neq t_1$, then we pop $p_i$ out of the stack $Q(p)$ of $p$, where $p=r(p_i)$. In this case, we do not need to update $\overline{t_1t_2}$. Otherwise, we also need to update $\overline{t_1t_2}$, by the deletion-type tangent searching procedure as before. Lemma~\ref{lem:costconvexhull} is still applicable here (replacing $L\cup R$ with $U$), so the procedure takes $O(1)$ amortized time.

If $p_i$ is the only point in $S_1$, then we do the following. We move $\ell$ to the rightmost point of $S_2$, and thus, the new set $S_1$ consists of all points in the old set $S_2$ while the new $S_2$ becomes $\emptyset$.
%Since $p_i$ is the only point in $S_1$, $p_{i+1}$ must be the leftmost point of $S_2$ and $\ell$ must contain $p_{i+1}$. Hence, $p_{i+1}$ is the leftmost vertex of $\calU(S_2)$. Let $p_j=r(p_{i+1})$, i.e., the right neighbor of $p_{i+1}$ on $\calU(S_2)$. We move $\ell$ to $p_j$ and set $t_1=p_{i+1}$ and $t_2=p_j$. It is easy to see that $\overline{t_1t_2}$ is indeed the common tangent of the upper hulls of the new sets $S_1$ and $S_2$ after $\ell$ is moved to $p_j$.  Clearly, $S_1=\{p_{i+1},p_{i+2},\ldots,p_{j-1}\}$.
Next, we build the list-stack data structure for $S_1$ by running Graham's scan leftwards from its rightmost point, which takes $|S_1|$ time. We call it the {\em left-hull construction procedure}. The following lemma implies that the amortized cost of the procedure is $O(1)$.

\begin{lemma}\label{lem:lefthullcon}
Every point of $U$ is involved in the left-hull construction procedure at most once in the entire algorithm.
\end{lemma}
\begin{proof}
Consider a point $p$ involved in the procedure. Notice that the procedure will not happen again before all points of the current set $S_1$ are deleted. Since $p$ is in the current set $S_1$, $p$ will not be involved in the same procedure again in future. \qed
\end{proof}

%The above describes our algorithm for handling a deletion, which takes $O(1)$ amortized time.

\subsection{Insertions}
\label{subsec:insertiongeneral}
Suppose a point $p_j$ is inserted. We first update $\calU(S_2)$ using Graham's scan, and we call it the {\em right-hull updating procedure}. After that, $p_j$ becomes the rightmost vertex of the new $\calU(S_2)$. The procedure takes $O(1+k)$ time, where $k$ is the number of vertices got removed from the old $\calU(S_2)$ (we say that these points are {\em involved} in the procedure). By the standard Graham's scan, the amortized cost of the procedure is $O(1)$. Note that although the line $\ell$ may move rightwards, we can still use the same analysis as the standard Graham's scan. Indeed, according to our algorithm for processing deletions discussed above, once $\ell$ moves rightwards, all points in $S_2$ will be in the new set $S_1$ and thus will never be involved in the right-hull updating procedure again in future.

After $\calU(S_2)$ is updated as above, we check whether the upper tangent $\overline{t_1t_2}$ needs to be updated, and if yes (in particular, if $S_2=\emptyset$ before the insertion), we run an insertion-type tangent searching procedure to find the new tangent in the same way as before in Section~\ref{sec:preconvexhull}. Lemma~\ref{lem:insertioncost} still applies (replacing $L\cup R$ with $U$), and thus the procedure takes $O(1)$ amortized time. This finishes the processing of the insertion, whose amortized cost is $O(1)$.

As a summary, in the general problem setting, we can perform each insertion and deletion in $O(1)$ amortized time, and after each update, the upper hull $\calU(S)$ can be output in $|\calU(S)|$ time.

\section{Convex hull queries}
\label{sec:queries}
In this section, we enhance the data structure described in Section~\ref{sec:generalconvexhull} so that it can support $O(\log h)$ time convex hull queries as stated in Theorem~\ref{theo:convexhull}, where $h$ is the number of vertices of the convex hull $CH(S)$.

The main idea is to use a red-black tree $T$ (or other types of finger search trees~\cite{ref:BrodalFi04}) to store the vertices of $\calU(S)$ in the left-to-right order with two ``fingers'' (i.e., two pointers) at the leftmost and rightmost leaves of $T$, respectively.\footnote{In the preliminary version of the paper at WADS 2023, we used interval trees and achieved $O(\log n)$ query time. In this version, following Michael T. Goodrich's suggestion, we resort to finger search trees instead, which help reduce the query time to $O(\log h)$ and also simplifies the overall algorithm.} During the course of the algorithm, $T$ will be updated with insertions and deletions. We will show that each insertion/deletion must happen at either the leftmost or the rightmost leaf, which takes $O(1)$ amortized time with the help of the two fingers~\cite{ref:TarjanAn88}. Using $T$, we can easily answer binary-search-based queries on $CH(S)$ in $O(\log h)$ time each.

Unless otherwise stated, we follow the same notation as those in
Section~\ref{sec:generalconvexhull}. In addition to the data structure for storing $S_1$ and $S_2$ described in Section~\ref{sec:generalconvexhull}, we assume that $T$ stores the vertices of $\calU(S)$ in the left-to-right order. In what follows, we discuss how to update $T$ during the algorithm in Section~\ref{sec:generalconvexhull}.

\subsection{Deletions}

Consider the deletion of a point $p_i$. As before, depending on whether $p_i$ is the only point of $S_1$, there are two cases.

\paragraph{\bf $\boldsymbol{p_i}$ is the only point of $\boldsymbol{S_1}$.}
If $p_i$ is the only point in $S_1$, then according to our algorithm,
we need to perform the left-hull construction procedure on the points $\{p_{i+1},p_{i+2},\ldots,p_{j}\}$, where $p_j$ is the rightmost point of $S_2$, after which the above set of points becomes the new $S_1$. In addition to the algorithm described in Section~\ref{subsec:deletegeneral}, we update $T$ as follows.

Recall that the left-hull construction procedure process vertices of $\{p_{i+1},p_{i+2},\ldots,p_{j}\}$ from right to left.
Suppose a vertex $p_k$ is being processed (i.e., points $p_{k+1},\ldots,p_j$ have already been processed and $\calU(S)$ is thus the upper hull of these points; $T$ store the vertices of $\calU(S)$). Then, using Graham's scan, we check whether the leftmost vertex $p_g$ of $\calU(S)$ needs to be removed due to $p_k$. If yes, we delete $p_g$ from $\calU(S)$; note that $p_g$ must be at the leftmost leaf of $T$. We continue this process until the leftmost vertex $p_g$ of the current $\calU(S)$ should not be removed due to $p_i$. Then,
we insert $p_k$ to $T$ as the leftmost leaf. Since each insertion and deletion of $T$ only happens at its leftmost leaf and we already have a finger there, each such update of $T$ takes $O(1)$ amortized time. Further, due to Lemma~\ref{lem:lefthullcon}, the amortized cost of the left-hull construction is still $O(1)$.

\paragraph{\bf $\boldsymbol{p_i}$ is not the only point of $\boldsymbol{S_1}$.}
Suppose $p_i$ is not the only point in $S_1$. Then we preform the same processing as before and update $T$ accordingly. There are two subcases depending on whether $p_i=t_1$.

If $p_i\neq t_1$, recall that our algorithm pops $p_i$ out of the stack $Q(p)$, where $p=r(p_i)$. In this case, the common tangent $\overline{t_1t_2}$ does not change. We update $T$ as follows. We first delete $p_i$ from $T$.  Observe that $p_i$ must be at the leftmost leaf of $T$. Therefore this deletion costs $O(1)$ amortized time. Let $\calU(S)$ refers to the upper hull after $p_i$ is deleted. Let $S'$ denote the set of vertices of $\calU(S)$ that were not on $\calU(S)$ before $p_i$ is deleted. It is not difficult to see that points of $S'$ are exactly the vertices of $\calU(S)$ strictly to the left of $p$. Starting from $p$, these vertices can be accesses from right to left using the list-stack data structure. We insert these points to $T$ in the right-to-left order. Observe that each newly inserted point becomes the leftmost leaf of the new tree (note that before these insertions, $p$ was at the leftmost leaf of $T$). Therefore, each such insertion on $T$ takes $O(1)$ amortized time. We refer to this process as {\em deletion-type tree update procedure}; we say that points of $S'$ are {\em involved} in this procedure. The following lemma implies that the amortized time of this procedure is $O(1)$ and so is the amortized time of deleting $p_i$.

\begin{lemma}\label{lem:deltreeupdate}
Each point can be involved in the deletion-type tree update procedure at most once in the entire algorithm.
\end{lemma}
\begin{proof}
Let $p'$ be any point of $S'$. We argue that $p'$ cannot be involved in the deletion-type tree update procedure again in future. Indeed, $p'$ was involved in the procedure due to the deletion of $p_i$. Since $p_i\neq t_1$, $p'$ must be a point in $S_1$. Right before $p_i$ is deleted, $p'$ was not on the upper hull of $S_1$. After $p_i$ is deleted, since no points will be inserted into $S_1$ (before $S_1$ becomes empty), $p'$ will always be on the upper hull of $S_1$ until it is deleted. This implies that $p'$ cannot be involved in the procedure again in future. \qed
\end{proof}

If $p_i= t_1$ (see Fig.~\ref{fig:convexhulldeletion}), recall that our algorithm finds the new tangent $\overline{t_1't_2'}$ using the deletion-type tangent searching procedure as described before. We update $T$ accordingly as follows. Observe that $p_i$, which is $t_1$, is the leftmost vertex of $\calU(S)$ and thus is at the leftmost leaf of $T$. We delete $p_i$ from $T$, which takes $O(1)$ amortized time. Then, we insert the points of $\calU(S_2)$ from $t_2$ to $t_2'$ in this order so that each newly inserted point becomes the leftmost leaf of $T$ (and thus each insertion takes $O(1)$ amortized time). Note that the above points can be accessed one by one using their left neighbor pointers, starting from $t_2$. Next, we insert $t_1'$ to $T$ at the leftmost leaf and then insert the vertices of $\calU(S_1)$ from $t_1'$ to its leftmost vertex (which is $p_{i+1}$) in this order; these points can be accessed one by one using the list-stack data structure, starting from $t_1'$. Again,  each newly inserted point becomes the leftmost leaf of $T$ and thus each such insertion takes $O(1)$ amortized time. Observe that each of the above newly inserted point $p_k$ of $T$ belongs to one of the following three cases: (1) $p_k$ is is a vertex of $\calU(S_1)$ strictly between $t_1'$ and $p_{i+1}$; (2) $p_k$ is is a vertex of $\calU(S_2)$ strictly between $t_2'$ and $t_2$; (3) $p_k$ is either $t_1'$ or $t_2'$. For Case (1), by the same analysis as in Lemma~\ref{lem:deltreeupdate}, $p_k$ will not be involved in this process again in future. For Case $(2)$, since $p_k$ is involved in the deletion-type tangent searching procedure (for computing $\overline{t_1't_2'}$), by Lemma~\ref{lem:costconvexhull}, $p_k$  will not be involved in this process again in future. Therefore, the amortized time of this process is $O(1)$, so is the amortized time of deleting $p_i$.

\subsection{Insertions}

Suppose we insert a point $q_j$. In addition to our processing as described in Section~\ref{sec:generalconvexhull}, we update $T$ as follows. Recall that our algorithm first runs a right-hull updating procedure on $\calU(S_2)$ by Graham's scan. During the procedure, if we need to remove a point $q_k$ from $\calU(S_2)$ and $q_k$ is to the right of $t_2$ (including the case $q_k=t_2$), then $q_k$ must be the rightmost vertex of the current $\calU(S)$ and thus is the rightmost leaf of $T$. In this case, we delete $q_k$ from $T$. If $q_k=t_2$, then the rightmost leaf of $T$ becomes $t_1$. Recall that in this case our algorithm will perform an insertion-type tangent searching procedure to find the new tangent point $t_1'$ on $\calU(S_1)$ (see Fig.~\ref{fig:convexhull-10}). During the tangent searching procedure, we keep deleting the rightmost leaf of $T$ until $t_1'$ is found, after which we insert $q_j$ to $T$ as the rightmost leaf. As such, each deleted point $q_k$ of $T$ as above belongs to one of the following three cases: (1) $q_k$ is a vertex of the original $\calU(S_2)$ before $q_j$ is inserted; (2) $q_k$ is a vertex of $\calU(S_1)$ strictly between $t'_1$ and $t_1$; (3) $q_k$ is either $t_2$ or $t_1$. For Case (1), as discussed in Section~\ref{subsec:insertiongeneral}, $q_k$ will not be involved in this procedure again in future. For Case (2), since $q_k$ is involved in the insertion-type tangent searching procedure, by Lemma~\ref{lem:insertioncost}, $q_k$ will not be involved in this procedure again in future. Since each of the above update operations on $T$ happens at either the leftmost or the rightmost leaf of $T$, the amortized cost of handing the insertion of $q_j$ is $O(1)$.

\section{Concluding remarks}
\label{sec:con}

We proposed a data structure to dynamically maintain the convex hull
of a set of points in the plane under window-sliding updates. Although
the updates are quite special, our result is particularly interesting
because each update can be handled in constant amortized time and
binary-search-based queries (both decomposable and non-decomposable)
on the convex hull can be answered in logarithmic time each. In
addition, the convex hull itself can be retrieved in time linear in its
size.
Also interesting is that our method is quite simple.
% except seemingly the lowest common ancestor data
% structure~\cite{ref:BenderTh00,ref:HarelFa84}. Harel and
% Tarjan~\cite{ref:HarelFa84} first gave a fairly complicated data
% structure, and later Bender and
% Farach-Colton~\cite{ref:BenderTh00} presented a much
% simpler one. Both data structures work for trees that are not
% necessarily binary. It turns out that a
% special property in our problem makes the lowest common ancestor
% queries quite easy to solve: our tree $T$ is a complete binary tree. In
% a complete binary tree, Harel and Tarjan~\cite{ref:HarelFa84}
% (see Section~3~\cite{ref:HarelFa84}) showed that a lowest common ancestor query can be easily solved by direct calculations. Therefore, we can use that approach for our
% problem. Furthermore, another observation in our problem makes the
% calculation even simpler: in each query for a
% pair of nodes $(u,v)$, both $u$ and $v$ are leaves of $T$. This
% observation simplifies the calculation of the method
% in~\cite{ref:HarelFa84}, e.g., the step for checking whether one of
% the two query nodes is an ancestor of the other is not needed any
% more; refer to Section~3 in~\cite{ref:HarelFa84} for details.

In the dual setting, the problem is equivalent to dynamically
maintaining the upper and lower envelopes of a set $S^*$ of lines in
the plane under insertions and deletions. The window-sliding
updates become inserting into $S^*$ a line whose slope is
larger than that of any line in $S^*$ and deleting from $S^*$ the line
with the smallest slope. Alternatively, the lower envelope
of $S^*$ may be considered as a parametric heap~\cite{ref:KaplanFa01}.
A special but commonly used type of parametric heaps is kinetic
heaps~\cite{ref:BaschDa99,ref:KaplanFa01}. With our result,
we can handle each window-sliding insertion and deletion on
kinetic heaps in $O(1)$ amortized time, and the ``find-min'' operation
can also be performed in $O(1)$ amortized time, in the same way as
those in~\cite{ref:BrodalDy02,ref:KaplanFa01}.

\section*{Acknowledgment}
The author would like to thank Joseph S.B. Mitchell for posing the question at WADS 2023 about whether the $O(\log n)$ query time in the preliminary version of the paper can be improved to $O(\log h)$, and thank Michael T. Goodrich for suggesting the idea of using finger search trees to achieve this.


%%%%%%%%%%%%
\bibliographystyle{plain}
\bibliography{reference}

%%%%%%%%%%%%%%%%%%%%%%%%%%%%%%%%%%%%%%%%%%%%%%%%%%%%%%%%%%%%%
%% APPENDICES
%%%%%%%%%%%%%%%%%%%%%%%%%%%%%%%%%%%%%%%%%%%%%%%%%%%%%%%%%%%%%
\appendix

%\newpage
%\appendix
%\section*{Appendix}

%\vspace{0.1in}
%\noindent
%{\bf Lemma \ref{lem:10}.}
%{\em
%For any $i\in [1,n-1]$, $I_i(\alpha)\geq I_{i+1}(\alpha)$,
%$I_i(\beta(a_2,b_2))\geq I_{i+1}(\beta(a_2,b_2))$,
%$I_i(\gamma)\leq I_{i+1}(\gamma)$, and $I_i(\delta)\leq I_{i+1}(\delta)$  (e.g., see Fig.~\ref{fig:shift}).
%}
%\vspace{0.06in}

\end{document}